\documentclass[submission,copyright,creativecommons]{eptcs}
\providecommand{\event}{EPTCS 2025} 

\usepackage{graphicx}
\usepackage{amsmath}

\usepackage{amssymb}
\usepackage{rotating}
\usepackage{floatpag}
\rotfloatpagestyle{empty}
\usepackage{amsthm}
\usepackage{graphicx}
\usepackage[tone,extra]{tipa}
\usepackage{flagderiv}

\usepackage{amsfonts}
\usepackage{amssymb}
\usepackage{amsmath}
\usepackage{epsfig}
\usepackage[english]{babel}

\begin{document}
\title{An Unconventional View on Beta-Reduction\\ in Namefree Lambda-Calculus}

\author{Rob Nederpelt
\institute{Eindhoven University of Technology\\ Eindhoven, The Netherlands}
\institute{Dept.\ of Math.\ and Comp.\ Sc.\ }
\email{r.p.nederpelt@tue.nl}
\and
Ferruccio Guidi
\institute{University of Bologna\\ Bologna, Italy}
\institute{Dept.\ of Comp.\ Sc.\ and Engin.\ (DISI)}
\email{ferruccio.guidi@unibo.it}
}
\def\titlerunning{A View on Namefree Beta-Reduction}
\def\authorrunning{Rob Nederpelt and Ferruccio Guidi}

\maketitle












\theoremstyle{plain}
\newtheorem{The}{Theorem}[section]
\newtheorem{Ste}[The]{Theorem}
\newtheorem{Cor}[The]{Corollary}
\newtheorem{Lem}[The]{Lemma}
\newtheorem{Not}[The]{Notation}
\newtheorem{Rem}[The]{Remark}
\newtheorem{Rems}[The]{Remarks}
\newtheorem{For}[The]{Formula}
\newtheorem{Rul}[The]{Rule}
\newtheorem{Con}[The]{Convention}
\newtheorem{Note}[The]{Note}
\newtheorem{Exp}[The]{Explanation}

\theoremstyle{definition}
\newtheorem{Def}[The]{Definition}
\newtheorem{Exa}[The]{Example}
\newtheorem{Exas}[The]{Examples}
\newtheorem{Proc}[The]{Procedure}




\begin{abstract}
Terms in the lambda-calculus can be represented as planar trees decorated with symbols for `abstraction' and `application', and having variables as leaves. In this paper, we concentrate on the {\em branches\/} of such trees, rather than on the trees themselves. We reformulate several well-known notions of beta-reduction in this view. In a natural manner, this reconsideration eventually leads to a new form of beta-reduction, being {\em expanding\/} -- in the sense that the reduction of term $t_1$ to term $t_2$ entails that the tree of~$t_1$ is a subtree of the tree of $t_2$.
\end{abstract}












\def\res{\mathop{\vert\grave{}}}
\def\rres{\mathop{\vert\hbox}}
\def\rrightarrow{\rightarrow \hspace{-.8em}
\rightarrow}
\def\ssubset{\subset \hspace{-0.5 em} \subset}
\def\insem{[ \hspace{-.23em} |}
\def\outsem{] \hspace{-.33em} |}
\def\ovi{\overline{\imath}}
\def\ovj{\overline{\jmath}}
\def\Rar{\Rightarrow}
\def\rar{\rightarrow}
\def\Lar{\Leftarrow}
\def\lar{\leftarrow}
\def\LRar{\Leftrightarrow}
\def\lrar{\leftrightarrow}
\def\ttT{{\tt True}~}
\def\ttF{{\tt False}~}
\def\isval{\hspace{.4ex} \stackrel{\it  val}{= \hspace{-.5ex} =} \hspace{.4ex}}
\def\grval{\hspace{.4ex} \mid \hspace{-.5ex} \stackrel{\it  val}{= \hspace{-.5ex} =} \hspace{.4ex}}
\def\klval{\hspace{.4ex} \stackrel{\it  val}{= \hspace{-.5ex} =} \hspace{-.5ex} \mid \hspace{.4ex}}
\def\rrar{\twoheadrightarrow}
\def\red{\rightarrow_{\beta}}
\def\rred{\twoheadrightarrow_{\beta}}
\def\led{\leftarrow_{\beta}}
\def\lled{\twoheadleftarrow_{\beta}}
\def\ded{\stackrel{\Delta}{\rightarrow}}
\def\dded{\stackrel{\raisebox{-0.5ex}{\scriptsize{$\Delta$}}}{\twoheadrightarrow}}
\def\cded{\stackrel{\Delta~}{\twoheadrightarrow_{\beta}}}
\def\conv{=_{\beta}}
\def\donv{\stackrel{\Delta}{=}}
\def\cdonv{\stackrel{\Delta~}{=_{\beta}}}
\def\cond{{\tt con}_{\Delta}}
\def\rightarrowd{\rightarrow_\delta}
\def\rrard{\rrar_\delta}
\def\rrarbd{\rrar_{\beta\delta}}

\def\crar{\stackrel{\bf c~~}{\rightarrow_{\beta_1}}}
\def\crrar{\stackrel{\bf c~~}{\rrar_{\beta_1}}}

\def\frar{\stackrel{\bf f~~}{\rightarrow_{\beta_1}}}
\def\frrar{\stackrel{\bf f~~}{\rrar_{\beta_1}}}

\def\undL{\underline{\rmL}}

\def\lmea{||\,}
\def\rmea{\,||_{\Delta}}
\def\seml{[\![}
\def\semr{]\!]}

\def\bbN{\mathbb{N}}
\def\bbZ{\mathbb{Z}}
\def\bbR{\mathbb{R}}
\def\bbNp{\bbN^+}
\def\isN{=_{\bbN}}
\def\isNp{=_{\bbN^+}}
\def\nisN{\not=_{\bbN}}
\def\nisNp{\not=_{\bbN^+}}

\def\lobar{\lambda \underline{\omega}}
\def\lamC{\lambda {\rm C}}
\def\lamD{\lambda {\rm D}}
\def\lamDu{\lambda {\rm D}_{0}}
\def\lamDn{\lambda {\rm D}_{\it n}}
\def\lamDp{\lambda {\rm D}^+}
\def\lamCp{\lambda {\rm C}^+}
\def\lamDz{\lambda {\rm D}_{0}}
\def\lamP{\lambda {\rm P}}

\def\subZ{_{\bbZ}}

\def\isdef{\stackrel{d}{=}}
\def\ddef{\,:=\,}
\def\pplus{\,+\,}
\def\ccirc{~\circ~}
\def\vvdash{~\vdash~}
\def\bott{\bot \hspace{-1.4ex} \bot}
\def\smallbott{\bot \hspace{-1.2ex} \bot}
\def\mmin{\text -}
\def\noi{\noindent\hspace{1.2ex}}
\def\noin{\noindent--\hspace{1.2ex}}
\def\nin{\noindent}
\def\inn{\!<\hspace{-1em}-\,}
\def\ins{\!<\hspace{-1em}-}
\def\inR{\!<\hspace{-4.2mm}-\,}
\def\inr{\!<\hspace{-3.5mm}-\,}
\def\inst{\stackrel{\raisebox{-3em}{--}}{<}}
\newcommand{\ina}{~\varepsilon~}

\def\inw{\in^{\wedge}}
\def\inv{\in_{\vee}}
\def\inwv{\in^{\wedge}_{\vee}}
\def\Tcar{{\cal T}^{\it car}}
\def\Tcarp{{\cal T}^{\it car+}}
\def\Tfre{{\cal T}^{\it fre}}
\def\Texp{{\cal T}^{\it exp}}
\def\Pexp{{\cal P}_{\it exp}}

\def\rmA{{\rm A}}
\def\rmL{{\rm L}}
\def\rmS{{\rm S}}
\def\rmP{{\rm P}}
\def\bft{{\bf t}}
\def\bfs{{\bf s}}
\def\bfu{{\bf u}}
\def\bfv{{\bf v}}


\def\sup{}

\newcommand{\be}{\beta}
\newcommand{\la}{\lambda}
\newcommand{\de}{\delta}
\newcommand{\si}{\sigma}
\newcommand{\ze}{\zeta}
\newcommand{\zei}{\zeta_{\star}}
\newcommand{\up}{\upsilon}
\newcommand{\eps}{\in}
\newcommand{\body}{{\tt body}}
\newcommand{\loc}{\tt loc}
\newcommand{\glo}{\tt glo}
\newcommand{\da}{\dagger}
\newcommand{\lamdot}{ \, . \,\, }
\newcommand{\ao}{\{}
\newcommand{\as}{\}}
\newcommand{\nat}{{\it nat}}
\newcommand{\bool}{{\it bool}}
\newcommand{\omegab}{{\underline{\omega}}}
\newcommand{\llabel}[1]{\label{#1}}

\newcommand{\cindex}[2]{\dindex{${#1}$,
Fig.\ref{#2}}}

\newcommand{\nindex}[1]{\index{author}{#1}}
\newcommand{\dindex}[1]{\index{definition}{#1}}
\newcommand{\tindex}[1]{\index{technot}{#1}}
\newcommand{\sindex}[1]{\index{subject}{#1}}

\newcommand{\textps}{\textprimstress}
\newcommand{\textss}{\textsecstress}
\newcommand{\calD}{{\cal D}}
\newcommand{\calR}{{\cal R}}
\newcommand{\calT}{{\cal T}}
\newcommand{\calE}{{\cal E}}
\newcommand{\calED}{{\calE_{\lambda D}}}
\newcommand{\calEDn}{{\calE_{\lamDn}}}
\newcommand{\calEDu}{{\calE_{\lamDu}}}
\newcommand{\calEDp}{\calE_{\lamDp}}
\newcommand{\calB}{{\cal B}}
\newcommand{\calV}{{\cal V}}
\newcommand{\calS}{{\cal S}}
\newcommand{\calP}{{\cal P}}

\newcommand{\calAn}{{\cal A}_0}
\newcommand{\calAo}{{\cal A}_1}
\newcommand{\calAi}{{\cal A}_i}
\newcommand{\ovx}{\overline x}
\newcommand{\ovy}{\overline y}

\newcommand{\bif}{{\sf if}}
\newcommand{\Zero}{{\sf Zero}}
\newcommand{\bthen}{{\sf then}}
\newcommand{\belse}{{\sf else}}
\newcommand{\Mult}{{\sf Mult}}
\newcommand{\Pred}{{\sf Pred}}

\newcommand{\List}{{\sf List}}
\newcommand{\IN}{{\mathbb N}}

\newcommand{\sss}{\scriptstyle\;}

\newcommand{\ovl}{\overline}

\newcommand{\ttp}{\tt p}
\newcommand{\ttTy}{\tt T}
\newcommand{\ttA}{\tt A}
\newcommand{\ttl}{\tt l}

\newcommand{\renlab}{\renewcommand{\labelitemi}{$-$}}
\newcommand{\listbegin}{\vspace{-1ex}\begin{itemize}\renlab}
\newcommand{\listend}{\end{itemize}\vspace{-1ex}}


\noindent{\em Dedicated to Stefano Berardi on the occasion of his 64th birthday in honor of his extensive and diverse research in theoretical computer science.}

\section{Preliminary remarks}
\label{SecPre}

\subsection{Introduction}
\label{SecInt}

It is well known that practical implementations of the $\lambda$-calculus turn $\alpha$-equivalence into syntactic equality by representing bound variable occurrences with depth indices or level indices, being positive numbers, rather than with names (\cite{deB72}). Therefore, the resulting systems are termed namefree as opposed to namecarrying.
Generally speaking, $\beta$-reduction involves replacing the occurrences of the bound variables in the body of the function with copies of the argument of the function, and, in this scenario, the indexes occurring in such copies may need an update to prevent captures.
Experience shows that this update, known as lift according to a well-established terminology, is a time consuming operation \cite[Appendix A2]{Gui09} that, precisely, computing machines strive to avoid (\cite{Klu05}). In a family of systems originating from \cite{deB78b} and $\lambda\sigma$ (\cite{Aba91}), $\beta$-reductions do not apply the update immediately, but store it in the copied terms by adding specific constructs to them. Thus, a computation can delay updates at will or apply them whenever is the case.



One of the namefree systems we present in this article
is based just on $\beta$-reduction at a distance,
an extension of $\beta$-reduction that has been introduced in (\cite{Ned73}).
Such a reduction relation allows, for example, not only \mbox{$\beta$-reduction} $K \equiv (\lambda x y \lamdot L) M N \rightarrow (\lambda y \lamdot L[x := M])N$, but also the variant $K \rightarrow (\lambda x \lamdot L[y := N])M$. See also (\cite{Reg92,Reg94}). 

In deviation of the usual notation, lambda terms 
are presented as the set of complete branches belonging to rooted trees that are composed of vertices and labeled edges. With this representation, one obtains a transparent view on matching pairs of abstraction and application, each of which pairs may generate a beta-reduction at a distance. This transparency considerably facilitates our discussion.

\subsection{About the tree structure of lambda terms}
\label{SecAbo}

The main motivation behind this paper is sheer curiosity. It is obvious that terms in (untyped or typed) $\lambda$-calculus have a tree structure, decorated with abstraction- and application-symbols -- say $\lambda$ and $@$, respectively. In the {\em namecarrying\/} versions, the leaves are variables, but in the {\em namefree\/} case these are positive natural numbers. See Figure~\ref{FigAboNam}\,$(i)$ for the namefree tree corresponding to the (untyped) term

\smallskip

$(\ast)$ $(\lambda ((\lambda   \lambda \, 2 \, 2) \lambda  2 ) 1) \lambda  1$,

\smallskip

\noindent this being the nameless version of the namecarrying term

\smallskip

$(\lambda x \lamdot (\lambda y \lamdot \lambda z \lamdot y \, y)(\lambda u \lamdot x)x)\lambda v \lamdot v$.

\smallskip

In this paper, we generally investigate this {\em namefree tree format\/} (first described in (\cite{deB72}); see also (\cite{deB77})).

\smallskip

 Our work has been inspired by the following question:

\begin{quote} {\em (Branch focus). Is it feasible to describe various reduction relations by concentrating not on the {\bf trees\/} but on the {\bf branches\/} of the trees.}
\end{quote}

\noindent For a linear representation of all branches in Figure~\ref{FigAboNam}\,$(i)$, see Example~\ref{ExaAboTre}, left hand side, 
In the present paper, we try to answer this question for various well-known forms of beta-re\-duc\-tion. This leads us to a new form of beta-reduction that has the property of expanding {\em\/} the tree under consideration, without any losses.

\smallskip


\begin{figure}[ht]

\begin{picture}(400,130)(-20,0)

\put(13,118){$(i)$}

\put(72,118){$@$}
\put(73,118){\line(-1,-1){10}}
\put(79,118){\line(1,-1){10}}
\put(58,104){$\lambda$}
\put(60.5,103){\line(0,-1){6}}
\put(90.5,103){\line(0,-1){5.5}}
\put(88,104){$\lambda$}
\put(57,88){$@$}
\put(89,89){$1$}
\put(42,73){$@$}
\put(73,74){$1$}
\put(28,59){$\lambda$}
\put(58,59){$\lambda$}
\put(28,44){$\lambda$}
\put(58,44){$2$}
\put(27,28){$@$}
\put(13,14){$2$}
\put(43,14){$2$}

\put(58,88){\line(-1,-1){10}}
\put(64,88){\line(1,-1){10}}

\put(43,73){\line(-1,-1){10}}
\put(49,73){\line(1,-1){10}}

\put(31,58){\line(0,-1){5.5}}
\put(60.5,58){\line(0,-1){5.5}}

\put(31,42.5){\line(0,-1){6}}

\put(28,28){\line(-1,-1){10}}
\put(34,28){\line(1,-1){10}}

\put(15,-3){\rm namefree tree of $(\ast)$}

\put(123,118){$(ii)$}
\put(113,119){\line(0,-1){100}}

\put(156,126){\line(0,-1){10}}

\put(151.5,108){$@$}

\put(152,107){\line(-1,-1){10}}
\put(161.5,107){\line(1,-1){10}}

\put(132,85){\rm application}

\put(156,66){\line(0,-1){10}}

\put(153,48){$\lambda$}

\put(156,47){\line(0,-1){10}}

\put(134,25){\rm abstraction}

\put(184,108){$\rightarrow$}

\put(184,48){$\rightarrow$}

\put(225,126){\line(0,-1){14}}

\put(225,111){\circle*{4}}

\put(225,112){\line(-1,-1){15}}
\put(225,112){\line(1,-1){15}}

\put(240,96){\circle*{4}}

\put(240,96){\line(0,-1){14}}

\put(206,103){$\rmA$}

\put(236,103){$\rmS$}

\put(225,66){\line(0,-1){14}}

\put(225,51){\circle*{4}}

\put(225,52){\line(-0,-1){14}}



\put(215,41){$\rmL$}


\put(276,118){$(iii)$}
\put(266,119){\line(0,-1){100}}

\put(355,109){\circle*{4}}
\put(340,123){\circle*{4}}
\put(325,109){\circle*{4}}
\put(325,93){\circle*{4}}
\put(325,64){\circle*{4}}
\put(310,78){\circle*{4}}
\put(295,64){\circle*{4}}
\put(295,49){\circle*{4}}
\put(295,34){\circle*{4}}

\put(355,93){\circle*{4}}
\put(340,78){\circle*{4}}
\put(325,49){\circle*{4}}
\put(310,19){\circle*{4}}
\put(280,19){\circle*{4}}

\put(340,123){\line(-1,-1){15}}
\put(340,123){\line(1,-1){15}}

\put(355,108){\line(0,-1){27}}
\put(325,108){\line(0,-1){14}}
\put(340,78){\line(0,-1){12}}
\put(325,63){\line(0,-1){26}}
\put(310,19){\line(0,-1){12}}
\put(280,19){\line(0,-1){12}}

\put(325,93){\line(-1,-1){15}}
\put(325,93){\line(1,-1){13}}

\put(310,78){\line(-1,-1){15}}
\put(310,78){\line(1,-1){15}}

\put(295,63){\line(0,-1){30}}
\put(325,63){\line(0,-1){12}}

\put(295,35){\line(-1,-1){15}}
\put(295,35){\line(1,-1){15}}

\put(346,83){$1$}
\put(331,68){$1$}
\put(316,39){$2$}
\put(271,9){$2$}
\put(301,9){$2$}

\put(324,116){$\rmA$}
\put(349,116){$\rmS$}

\put(315,99){$\rmL$}
\put(345,99){$\rmL$}

\put(309,86){$\rmA$}
\put(334,86){$\rmS$}

\put(294,71){$\rmA$}
\put(319,71){$\rmS$}

\put(285,54){$\rmL$}
\put(315,54){$\rmL$}

\put(285,39){$\rmL$}

\put(278,26){$\rmA$}
\put(305,26){$\rmS$}

\put(272,-3){\rm namefree $\lambda$-tree of $(\ast)$}

\end{picture}
\caption{Namefree lambda trees; traditional and adapted}
\label{FigAboNam}
\end{figure}

Our first goal is to ensure that {\em the entire set of branches represents the tree from which it originates.\/} This brings about several issues having to be considered:

\smallskip

$(i)$ {\em sound tree reconstruction\/} ~ A tree as in Figure~\ref{FigAboNam}\,$(i)$ is {\em planar\/}, i.e., below every $@$ follow a {\em left\/}  and a {\em right\/} branch. However, in {\em branches\/} containing an $@$ there is no clue whether the original path in the tree went left or right.

We have chosen to repair this by adding a 'label' $\rmS$ (for `subterm') on top of the {\em right\/} branch, leaving the left branch unchanged. See Figure~\ref{FigAboNam}\,$(ii)$. For an easier representation, we use $\rmA$ instead of $@$, and $\rmL$ instead of $\lambda$. This will be done henceforth in the paper.

\smallskip

$(ii)$ {\em sound redex detection\/} ~ In the tree, an $@$ with a $\lambda$ immediately {\em left\/} below it, determines a redex. However, an $@$ with a $\lambda$ immediately {\em right\/} below it does {\em not\/} (cf Example~\ref{ExaAboTre} (3)).

The proper addition of labels $\rmS$, as described right now, ensures that the difference between `descending to the left' and `descending to the right' has been covered by whether, no or yes, there is an $\rmS$ between $@$ and $\lambda$.


\smallskip

$(iii)$ {\em type preparedness\/} ~ The standard procedure for the detection of the $\lambda$ binding a (numeric) variable $n$ in a tree is this: follow (upward) the branch ending in this $n$ and subtract 1 for every $\lambda$ met. The $\lambda$ where 1 changes to 0 is the binding $\lambda$ for $n$. (See Figure~\ref{FigAboNam}\,$(i)$ for examples.)

Preferably, our focus on branches should be appropriate for the extension to {\em typed\/} lambda calculus 
(cf (\cite{Bar92})). However, there is an annoying anomaly in the described search for the binding $\lambda$. To be precise: the count from $n$ to 0, upward along the branch, should {\em bypass\/} every $\lambda$ at a bifurcation that is approached from {\em right\/} below. The reason is that the subtree right below such a $\lambda$ (representing a type) does not contain numeric variables bound by this~$\lambda$. This has been explicitly noted in (\cite{deB87}, Section~2.) We illustrate this anomaly in Figure~\ref{FigAboBin}\,$(i)$.

\begin{figure}[ht]

\begin{picture}(400,77)(21,75)

\put(103,141){$(i)$}

\put(135,110){$\times!$}

\multiput(142,146)(0,-5){3}{\circle*{1.5}}
\multiput(147,125)(5,-5){3}{\circle*{1.5}}
\put(137.2,125.2){\line(-1,-1){10}}
\put(138,126.5){$\lambda$}
\put(123,108){$\lambda$}
\put(139,88.9){$1$}
\multiput(112.5,93)(0,-5){3}{\circle*{1.5}}

\put(122,107){\line(-1,-1){10}}
\put(131.5,107){\line(1,-1){10}}

\put(145,92.5){\circle{2}}
\multiput(145,97)(-3.3,3.3){5}{\circle{2}}
\multiput(129.5,113.3)(2.90,2.90){5}{\circle{2}}

\put(149,68){\rm typed abstraction}

\put(191,108){$\rightarrow$}

\put(217,141){$(ii)$}

\multiput(230,91)(0,-5){3}{\circle*{1.5}}
\multiput(260,146)(0,-5){4}{\circle*{1.5}}

\put(267.4,89.8){\circle{2}}
\multiput(267.4,94.1)(-3.1,3.1){4}{\circle{2}}
\multiput(251.5,109.5)(-3.1,3.1){4}{\circle{2}}
\put(245,121.1){\circle{2}}

\put(260,127){\line(-1,-1){30}}
\put(245,112){\line(1,-1){15}}
\put(260,127){\line(1,-1){15}}


\put(260,96){\line(0,-1){14}}

\put(231,106){$\rmL$}
\put(246,121){$\rmL$}

\put(253,106){$\rmS$}
\put(268,121){$\rmS$}
\put(261,86){$1$}


\put(230,96){\circle*{4}}
\put(260,126){\circle*{4}}
\put(260,96){\circle*{4}}
\put(275,111){\circle*{4}}



\end{picture}
\caption{The binding of a variable; traditional and adapted}
\label{FigAboBin}
\end{figure}
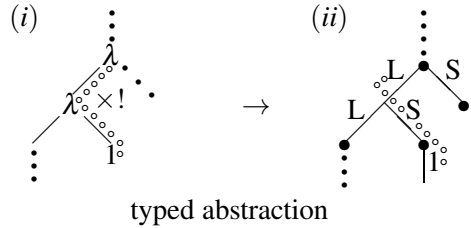

$(iv)$ {\em skipping a useless $@$-sign\/} ~ A similar remark as in $(iii)$ holds, in both the untyped and the typed cases, when not a $\lambda$, but an $@$ is approached from right below. Such an $@$ should actually be also skipped in the 'coding' of the branch, since this $@$ has no influence whatsoever.

This is awkward. Our solution to $(iv)$ (and also $(iii)$!) is to lower all labels $@$ and $\lambda$, in the sense that they become attached to the {\em edge left below\/} the original label (so not to the vertex). Moreover, we attach the new label $\rmS$ to the edge {\em right below\/} the bifurcation. Apart from that, we add an edge for every numeric variable, and attach the variable to this new edge.
See Figure~\ref{FigAboNam}\,$(ii)$ and also Figure~\ref{FigAboBin}\,$(ii)$.

For the skipping of an $@$-sign, see Example~\ref{ExaAboTre} (2), (3) and (4): compare the traditional branches and the adapted ones.

\begin{Exa}\label{ExaAboTre}
The trees in Figure~\ref{FigAboNam} have five branches. The lists of the labels, from root to leaf, along these branches, are the following. The branches of the tree are ordered from left to right.

\medskip

\begin{tabular}{ll|l}

 & \rule[-0.7em]{0em}{0.7em}In the traditional tree:~~ & ~~In the adapted tree: \\

(1) & $ @ \, \lambda \, @ \, @ \, \lambda \, \lambda \, @ \, 2 $
 & ~~$ \rmA \, \rmL \, \rmA \, \rmA \, \rmL \, \rmL \, \rmA \, 2 $ \\

(2) & $ @ \, \lambda \, @ \, @ \, \lambda \, \lambda \, @ \, 2 $
 & ~~$ \rmA \, \rmL \, \rmA \, \rmA \, \rmL \, \rmL \, \rmS \, 2 $ \\

(3) & $ @ \, \lambda \, @ \, @ \, \lambda \, 2 $
 & ~~$ \rmA \, \rmL \, \rmA \, \rmS \, \rmL \, 2 $ \\

(4) & $ @ \, \lambda \, @ \, 1 $
 & ~~$ \rmA \, \rmL \, \rmS \, 1 $ \\

(5) & $ @ \, \lambda \, 1 $
 & ~~$ \rmS \, \rmL \, 1 $
 \end{tabular}
\end{Exa}

\begin{Note}\label{NotAboIde}
Branches (1) and (2) are {\em identical\/} in the traditional tree, but {\em different\/} in the adapted tree.
\end{Note}

\begin{Note}\label{NotAboExp}
Another advantage of the tree representation is, that there is no need for an extra marker or another technical intervention to delimit a subterm, as is often required in the process of {\em explicit substitution} being executed on a {\em linear} presentation of a $\lambda$-term. For example, let's consider an explicit substitution operator, say $\Sigma$, that we want to `propel' one step forward through a linearly written \mbox{$\lambda$-term}: $\ldots \Sigma\,{\bf t_1}\, {\bf t_2} \ldots$. We assume that subterm ${\bf t_1}\,{\bf t_2}$ is written as function $\bf t_2$ {\em preceded\/} by argument $\bf t_1$; cf., (\cite{deB72}). Then the result of the propelling of $\Sigma$ could look like $\ldots  ({\Sigma \, \bf t_1 \,\sharp}) {\, \Sigma\, \bf t_2} \ldots$, in which the inserted symbol $\sharp$ delimits the subterm $\bf t_1$, so that the first copy of $\Sigma$ can halt in time (and becomes erased).
Cf., (\cite{Ned79}, p.\ 7) and (\cite{Ned80}, p.\ 5, 6).
\end{Note}

\subsection{Lambda trees and paths}\label{SecLam}

We give the name {\em $\lambda$-tree\/} to trees of lambda terms as exemplified in the {\em adapted\/} tree of Figure~\ref{FigAboNam}\,$(iii)$. See Definition~\ref{DefLamTre}, $(i)$, below. For $\lambda$-trees, we use the word {\em path\/} for a branch or a part of a branch and the word {\em num{\rm -}label\/} or {\em num{\rm -}variable\/} for a numeric variable.


\begin{Def}\label{DefLamTre}

$(i)$ A {\em lambda-tree\/} is a connected acyclic undirected graph, ranged over by $\bft$, $\bft'$, \ldots, constructed by the inductive definition below. Such a tree must be non-empty, rooted, and edge-labeled.

Let $\bft$, $\bft_1$ and $\bft_2$ be $\lambda$-trees and $n$ a positive natural number. Then also the following trees are $\lambda$-trees:

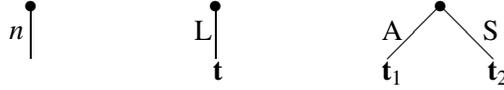
\begin{picture}(400,50)(-40,-10)

\put(35,10){\line(0,1){20}}
\put(27,17.5){$n$}
\put(35,30){\circle*{4}}

\put(97,17){$\rmL$}
\put(105,30){\circle*{4}}
\put(105,10){\line(0,1){20}}
\put(104,2){$\bft$}

\put(168,17){$\rmA$}
\put(190,30){\circle*{4}}
\put(170,10){\line(1,1){20}}
\put(168,2){$\bft_1$}

\put(206,17){$\rmS$}
\put(210,10){\line(-1,1){20}}
\put(208,2){$\bft_2$}

\end{picture}


\noindent ({\em Note:} For technical reasons we draw num-labeled edges with a `loose' end, not having a lower node. This makes it easy to extend a variable-labeled edge with a new lambda-tree, as we do in Section~\ref{SecNewLos}.)

$(ii)$ A {\em label\/} is one of $\rmA$, $\rmL$, $\rmS$ or any $n \in \bbNp$, ranged over by $\ell$,  \ldots.

$(iii)$ A {\em path\/} in a tree $\bft$  is a connected string of labeled edges occurring in $\bft$, recorded downwards. (A~path may identify only a {\em part\/} of a branch.) The paths are ranged over by $p$, $q$, \ldots.

\end{Def}

\begin{Lem}\label{LemLamDif}
Along different paths in a $\lambda$-tree, one finds different strings of labels.
\end{Lem}
{\em Proof\/} ~~  Let $p_1$ and $p_2$ be different paths in $\bft$. Find the leftmost position where the two paths deviate. This must be at a bifurcation. So, in that position, one label is an $\rmA$ and the other an $\rmS$. Hence, the strings of labels along $p_1$ and $p_2$ also differ. $\Box$

\smallskip

This lemma enables us to {\em identify\/} a path with its string of labels. See also Note~\ref{NotAboIde} and Definition~\ref{DefLamIde}. We give names to special types of paths.

\begin{Def}\label{DefLamPat}

Let ${\bf t}$ be a $\lambda$-tree and $p$ be a non-empty path in ${\bf t}$.
Notation: $p \in \bft$ and $p \not \equiv \varepsilon$.

$p$ is a {\em root path\/} of $\bft$ if $p$ starts in the root of $\bft$. Notation: $p \in^{\wedge} \bft$.

$p$ is a {\em leaf path\/} if it has a leaf as final label. Notation: $p \in_{\vee} \bft$.

$p$ is {\em complete\/} if it is both a root path and a leaf path. Notation: $p \in^{\wedge}_{\vee} {\bf t}$.
\end{Def}

Examples of complete paths: see Example~\ref{ExaAboTre}, right hand side.


\begin{Def}\label{DefLamIde}
Let $p_1$, \ldots, $p_n$ be all complete paths in a $\lambda$-tree ${\bf t}$. Then we  identify ${\bf t}$ with the set $\{p_1, \ldots, p_n\}$.
\end{Def}

\begin{Def}\label{DefLamLen}
$(i)$ The {\em lenght\/} $| p |$ of a path $p$ in the $\lambda$-tree $\bft$ is the number of labels (including num-labels) in $p$.

$(ii)$ The {\em {\rm L}-length\/} $\lVert p \lVert$ of a path $p \in \bft$ is the number of labels ${\rm L}$ that occur in~$p$.
\end{Def}

\smallskip

In {\it namefree\/} $\lambda$-calculus, the {\em binding of variables\/} -- as known from namecarrying lambda-calculus -- is expressed by the {\em value\/} of the num-label $n$ at the end of a complete path. This label $n$ represents a variable. The procedure for establishing the bindings between labels $\rmL$ and num-variables  has been discussed already in the previous section, under the heading {\em type-preparedness\/}. We can now give a simple definition of binding. (Note: path $p \, \rmL \, q \, n$ is the concatenation of path $p$, label $\rmL$, path $q$ and label $n$.)


\smallskip

The usual notions `bound' and `closed' in lambda-calculus are covered by the following definition.

\begin{Def} \label{DefLamBou} $(i)$ Let $\bft$ be a $\lambda$-tree and
$p \, \rmL \, q \, n \inwv \bft$ such that $n = \lVert q \lVert + 1$. Then this $n$ is {\em bound \/} by the mentioned $\rmL$. Moreover, the path ${\rm L} \, q \, n$ is called the {\em {\rm L}-block} of (this occurrence of) $n$.

$(ii)$
The $\lambda$-tree $\bft$ is {\em closed\/} if all num-variables in $\bft$ are bound by some $\rmL \in \bft$.

\end{Def}

Note that the binding $\rmL$ of an occurrence of a num-variable $n$ always occurs in the (unique) complete path ending in this occurrence of $n$.

\begin{Lem}\label{LemLamLbl} $(i)$ The ${\rm L}$-block of a certain $n \in \bft$, if it exists, is unique.

$(ii)$ In a closed term, each occurrence of a num-variable $n$ corresponds to exactly one ${\rm L}$-block; but even when the term is closed, not every ${\rm L}$-block binds some num-variable.
\end{Lem}

\medskip

To every $\lambda$-tree belongs a well-defined set $\calS$ of complete paths. A natural question is: when does a given set $\calS$ of paths define a $\lambda$-tree that can be constructed according to Definition~\ref{DefLamTre}?

Since this definition is inductive, it is no surprise that a direct procedure for deciding this question is inductive, as well. We now give a verbal representation of such a procedure.

\begin{Proc}\label{ProLamTre}
Firstly, we require that an arbitrary path in $\calS$ consists of elements of $\{ \rmL, ~ \rmA, ~ \rmS \}$ only, with, as an exception, the final label of such a path, which must be a positive natural number. We call such a path a {\em proper\/} path.

So we may assume that $\calS$ consists of proper paths. We further assume that $\calS$ is finite and that all paths in $\calS$ are finite, as well. In order to simplify the description of the procedure that we give in the following, we assume that the paths are lexicographically ordered, on a basic order, say $\rmL < \rmA < \rmS < 1 < 2 < \ldots$. We number the paths accordingly: $p_1, p_2, \ldots, p_n$.

Here comes the procedure for such a set $\calS$ of proper paths:

\underline{case 1:} Let $p_1 \equiv \rmL ~ q_1$ for some $q_1$. Then

{\em Requirement 1\/} ~ For all $1 \leq i \leq n:
~p_i \equiv \rmL ~ q_i$ for some $q_i$.

Skip the front-$\rmL$`s in all these paths; then we get $\calS' \equiv \{q_1, q_2, \ldots q_n \}$. Apply the procedure to~$\calS'$.


\underline{case 2:} Let $p_1 \equiv \rmA ~ q_1$ for some $q_1$. Then

{\em Requirement 2\/} ~ There must be some $p_i$ such that $p_i \equiv \rmS \, q_i$.

By the lexicographical ordering, there must now be an $1 \leq m \leq n$ such that all $p_k$ with $k \leq m$ begin with $\rmA$, and all $p_k$ with $k > m$ begin with $\rmS$.

Divide the set $\calS$ into two parts:

$\calS_1 := \{ p_k \in \calS ~|~ p_k {\rm ~begins~with~} \rmA \}$, and $\calS_2 := \{ p_k \in \calS ~|~ p_k {\rm ~begins~with~} \rmS \}$.

Skip the front-$\rmA$`s in all paths of $\calS_1$ and the front-$\rmS$`s in $\calS_2$ and collect them. We obtain $\calS'_1$ and $\calS'_2$. Apply the procedure to $\calS'_1$ and $\calS'_2$.


\underline{case 3:} Let $p_1 \equiv n$ for some positive $n$. Then

{\em Requirement 3\/} ~ $\calS \equiv \{ p_1 \}$.

\smallskip

If one of the requirements is not met, we abort the procedure and give the answer `no'. It is not hard to show that this procedure ends without abortion (and the answer is `yes') if and only if the original $\calS$ is the set of all paths belonging to one specific $\lambda$-tree. $\Box$

\end{Proc}

\section{beta-reduction}
\label{SecBet}

\subsection{A short history of updating in namefree beta-reduction}\label{SecSho}

In namecarrying systems of $\lambda$-calculus
a binder in a term $M$, say $\lambda x$, and the variable occurrences
that refer to it carry the same name, say $x$.
 In contrast, namefree systems use unnamed binders, say $\lambda$,
and replace a bound variable occurrence $x$
with an index that is a non-negative integer denoting the position of
the corresponding $\lambda x$ along the path connecting $x$ to the
root of $M$ in the representation of $M$ as an abstract syntax tree.
As we pointed out in the introduction,
the $\beta$-reduction step of the latter systems
requires updating the indexes occurring in a copied
argument, say $N$, to maintain the relationship between the bound variable
instances in $M$ and the respective binders.
Depending on the particular system,
if {\em immediate updating\/} is in effect, the update occurs by
applying a so-called {\em update function\/} to the indexes in $N$. In contrast, if {\em delayed updating\/} is in effect, the update function
is just stored in the syntax of the copied $N$.

The first namefree systems with immediate updating appear in (\cite{deB72})
with the basic update functions $\tau_{d,h}$ of type $\bbNp \to \bbNp$,
where $d\in\bbN$ and $h\in\bbN$.
\[
\tau_{d,h} \equiv i \mapsto
\left\{\begin{tabular}{ll}
$i$&if $i \le d$\\
$i+h$&if $i > d$\\
\end{tabular}\right.
\]
The systems accompanying (\cite{deB78b})
-- for example, those of (\cite{deB77,deB78a}) --
are the first to allow delayed updating
by featuring the term node $\phi(f)$
where $f$ is an arbitrary function of type $\bbNp \to \bbNp$. The original purpose of $\phi(f)$ is to present substitution as a single operation defined by recursion on the structure of terms.

Other systems of the same family, such as (\cite{Ned79,Ned80}), (\cite{KN93}),
feature the term node $\mu(d,h)$ or $\phi^{(d,h)}$
that holds the function $\tau_{d,h}$.
Moreover, the systems originating from (\cite{Aba91})
-- for instance those in (\cite{CHL96}, 1996) --
feature the explicit substitution constructors $\mathit{id}$ and $\uparrow$
that essentially hold the functions $\tau_{0,0}$ (the identity)
and $\tau_{0,1}$ (the successor) respectively.

\subsection{The usual beta-reduction in the path-approach}\label{SecUsu}

We continue with a number of useful definitions for namefree $\lambda$-calculus with the emphasis on paths.

\begin{Def}\label{DefUsuGra}
Let ${\bf t}$ be a $\lambda$-tree and $p$ a fixed root path in ${\bf t}$.

The set of  all paths $q \inv \bft$ such that $p \, q$ is a complete path in $\bft$, is denoted by ${\it tree}(p)$.
We call the set $\{p \, q \, | \, q \inv {\it tree}(p)\}$ the {\em grafted tree\/} of $p$ in ${\bf t}$.
\end{Def}

We note that the grafted tree of a root path $p$ in a {\em closed\/} $\bft$ is  `closed' itself, in the sense that all free variables in ${\it tree}(p)$ are bound in $p$.

We shall now describe the usual $\beta$-reduction in terms of paths and grafted trees. We start with the well-known notion `redex' (reducible expression).

\begin{Def}\label{DefUsuRed}

Let $\bft$ be a $\lambda$-tree. Let $p \, \rmA \, \rmL \, \inw \bft$. Then the adjacent pair $\rmA \, \rmL$ at the end of this path identifies a {\em redex\/}. This redex consists of two elements: (1) the {\em `function'\/} $\rmL \, \, {\it tree}(p \, \rmA \, \rmL)$ and (2) the {\em `argument'\/} ${\it tree}(p \, \rmS)$.
\end{Def}
See Figure~\ref{PicBetRed},~$(i)$.

\smallskip

We now consider the usual relation called $\beta$-reduction and expressed with the symbol $\red$. This $\beta$-reduction formalizes the action: `apply a function to an argument'. 
In namefree lambda-calculus, which is our subject here, $\beta$-reduction has important consequences for the numbers acting as num-variables. Some of these numbers should be `updated' after the $\beta$-reduction.



\begin{Def}\label{DefUsuBet} Let $\bft_1$ be a $\lambda$-tree. Assume that $p \, \rmA \, \rmL \inw \bft_1$. (For reference, we call the $\rmA$ and the $\rmL$ in this path {\em pivotal\/}.)
Consider the corresponding grafted tree $p \, \rmA \, \rmL \, {\it tree}(p \, \rmA \, \rmL)$.
Then $\bft_1 \red \bft_2$, where $\bft_2$ is the tree obtained from $\bft_1$ by

\smallskip

$(i)$ substituting and {\it updating\/} (see below) ${\it tree}(p \, \rmS)$ for every num-variable in ${\it tree}(p \, \rmA \, \rmL)$ that is bound by the pivotal $\rmL$,

$(ii)$ erasing the pivotal $\rmA$-$\rmL$-pair, and

$(ii)$ erasing all complete paths in the grafted tree $p \, \rmS \, {\it tree}(p \, \rmS)$.

\end{Def}

\begin{Def}\label{DefUsuUpd}
{\em Updating\/} num-variables due to $\beta$-reduction is the process illustrated in Figure~\ref{PicBetRed}. In picture $(ii)$ of this Figure, we distinguish the following cases regarding picture $(i)$:

$(1)$ $n < \lVert q \lVert + 1$,

$(2)$ $n > \lVert q \lVert + 1$,

$(3)$ $n = \lVert q \lVert + 1$, separated by subcases $l \leq \lVert r \lVert$ and $l > \lVert r \lVert$,

\end{Def}

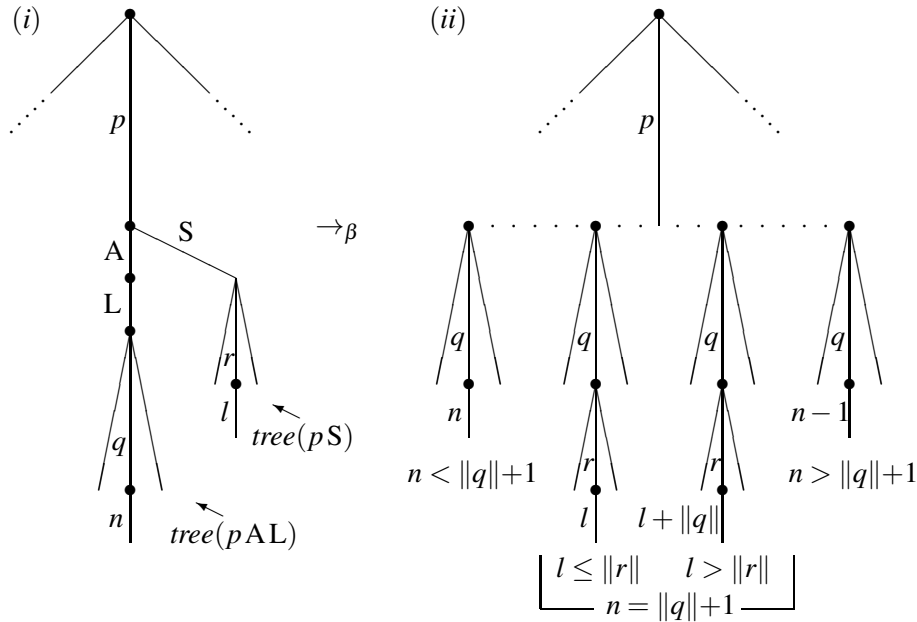
\begin{figure}[ht]

\begin{picture}(250,240)(0,45)
\put(5,265){$(i)$}

\put(50,70){\line(0,1){200}}
\put(50,270){\line(-1,-1){30}}
\put(50,270){\line(1,-1){30}}
\multiput(17,237.3)(-3,-3){5}{\circle*{1}}
\multiput(83,237.3)(3,-3){5}{\circle*{1}}

\put(50,190){\line(2,-1){40}}
\put(90,170){\line(0,-1){60}}

\put(90,170){\line(-1,-5){8}}
\put(90,170){\line(1,-5){8}}

\put(50,150){\line(-1,-5){12}}
\put(50,150){\line(1,-5){12}}

\put(50,90){\circle*{4}}
\put(50,150){\circle*{4}}
\put(50,170){\circle*{4}}
\put(50,190){\circle*{4}}
\put(50,270){\circle*{4}}
\put(90,130){\circle*{4}}

\put(42,76){$n$}
\put(43,105){$q$}
\put(40,157){$\rmL$}
\put(40,177){$\rmA$}
\put(42,227){$p$}
\put(68,184){$\rmS$}
\put(85,137){$r$}
\put(84,116){$l$}

\put(96,107){${\it tree}(p \, \rmS)$}
\put(65,70){${\it tree}(p \, \rmA \,  \rmL)$}

\put(114,117){\vector(-2,1){10}}
\put(85,80){\vector(-2,1){10}}

\put(120,187){$\red$}

\put(163,265){$(ii)$}

\put(250,270){\line(0,-1){80}}
\put(226,190){\line(0,-1){120}}
\put(274,190){\line(0,-1){120}}
\put(250,270){\line(-1,-1){30}}
\put(250,270){\line(1,-1){30}}

\multiput(217,237.3)(-3,-3){5}{\circle*{1}}
\multiput(283,237.3)(3,-3){5}{\circle*{1}}

\put(178,190){\line(-1,-5){12}}
\put(178,190){\line(1,-5){12}}
\put(178,110){\line(0,1){80}}

\put(226,190){\line(-1,-5){12}}
\put(226,190){\line(1,-5){12}}
\put(226,130){\line(-1,-5){8}}
\put(226,130){\line(1,-5){8}}

\put(274,190){\line(-1,-5){12}}
\put(274,190){\line(1,-5){12}}
\put(274,130){\line(-1,-5){8}}
\put(274,130){\line(1,-5){8}}

\put(322,190){\line(-1,-5){12}}
\put(322,190){\line(1,-5){12}}
\put(322,110){\line(0,1){80}}

\put(178,130){\circle*{4}}
\put(178,190){\circle*{4}}

\put(226,90){\circle*{4}}
\put(226,130){\circle*{4}}
\put(226,190){\circle*{4}}

\put(250,270){\circle*{4}}

\put(274,90){\circle*{4}}
\put(274,130){\circle*{4}}
\put(274,190){\circle*{4}}

\put(322,130){\circle*{4}}
\put(322,190){\circle*{4}}

\put(170,116){$n$}
\put(171,145){$q$}
\put(219,145){$q$}
\put(221,97){$r$}
\put(220,76){$l$}
\put(267,145){$q$}
\put(269,97){$r$}
\put(242,76){$l + \Vert q \lVert$}
\put(242,227){$p$}
\put(300,116){$n - 1$}
\put(315,145){$q$}

\multiput(187,190)(8,0){8}{\circle*{1}}
\multiput(313,190)(-8,0){8}{\circle*{1}}


\put(155,93){$n <  \lVert q \lVert  + 1$}
\put(299,93){$n >  \lVert q \lVert  + 1$}
\put(211,57){$l \leq \lVert r \lVert$}
\put(260,57){$l > \lVert r \lVert$}
\put(230,43){$n =  \lVert q \lVert  + 1$}

\put(225,45){\line(-1,0){20}}
\put(281,45){\line(1,0){20}}
\put(205,45){\line(0,1){20}}
\put(301,45){\line(0,1){20}}
\end{picture}

\caption{A picture of namefree $\beta$-reduction with updating}
\label{PicBetRed}
\end{figure}

\begin{The}\label{TheUsuUpd}
Updating preserves the bond between num-variables and their binding $\rmL$-labels.
\end{The}

{\it Proof}~~We illustrate what happens under $\beta$-reduction in Figure~\ref{PicBetRed}. We show in part~$(i)$ of that figure the essential parts of the redex.
In part~$(ii)$ of the same picture we show how the updating works.

We justify the preservation of the bindings in the update process as follows. Here, for easy reference, we denote an $\rmL$ binding $n$ as $\rmL_n$ and an $\rmL$ binding $l$ as $\rmL_l$. We also use the symbol $\red$ in an unorthodox manner. We write $m^{\rm upd}$ for an updated $m$.

It suffices to inspect two representative paths:

\smallskip

$(i)$ in $p \, \rmA \, \rmL \, \, {\it tree}(p \, \rmA \, \rmL)$ we choose $p \, \rmA \, \rmL \, q \, n$, with $q \, n$ a generic path in ${\it tree}(p \, \rmA \, \rmL)$,

$(ii)$ in $p \, \rmS \, \, {\it tree}(p \, \rmS)$ we choose $p \, \rmS \, r \, l$, with $r \, l$ a generic path in ${\it tree}(p \, \rmS)$.

\smallskip

We discern the cases for $n$ as described in Definition~\ref{DefUsuUpd}:

(1) $n < \lVert q \lVert + 1$.
Then $p \, \rmA \, \rmL \, q \, n = p \, \rmA \, \rmL \, q_1 \, \rmL_n \, q_2 \, n \red p \, q_1 \, \rmL_n \, q_2 \, n^{\rm upd}$, and hence $n^{\rm upd} = n$.

(2) $n > \lVert q \lVert + 1$.
Then $p \, \rmA \, \rmL \, q \, n = p_1 \, \rmL_n \, p_2 \, \rmA \, \rmL \, q \, n \red p_1 \, \rmL_n \, p_2 \, q \, n^{\rm upd}$, and hence $n^{\rm upd} = n -1$ (since the pivotal $\rmL$ has been erased).

(3) $n = \lVert q \lVert + 1$.
Then the pivotal $\rmL$ binds to $n$. Now we have to distinguish two cases for $l$:

(3a) $l \leq \lVert r \lVert$.
Then $p \, \rmS \, r \, l = p \, \rmS \, r_1 \, \rmL_l \, r_2 \, l$, hence $p \, \rmA \, \rmL \, q \, n \red p \, q \, r_1 \, \rmL_l \, r_2 \, l^{\rm upd}$, so $l^{\rm upd} = l$.

(3b) $l > \lVert r \lVert$.
Then $p \, \rmS \, r \, l = p_1 \, \rmL_l \, p_2 \, \rmS \, r \, l$, hence $p \, \rmA \, \rmL \, q \, n \red p_1 \, \rmL_l \, p_2 \, q \, r \, l^{\rm upd}$, so $l^{\rm upd} = l + \lVert q \lVert$, since $q$ now appears between $\rmL_l$ and $r$. ~~$\Box$

\subsection{Comparing beta-reduction in namecarrying and namefree lambda-calculus}\label{SecCom}

The set of terms of the $\lambda$-calculus that we have exposed until now, with a focus on paths, we denote as $\Tfre$. In the present section, we compare it with the namecarrying $\lambda$-calculus with the same focus on paths, that we call $\Tcar$. We do not explain how the terms in $\Tcar$ look like. We assume that the reader can easily devise that. The most important differences with $\Tfre$ are:

(1) $\Tcar$ has actual variables (such as $x$, $y$\ldots) instead of num-variables.

(2) Every $\rmL$-label in $\Tcar$ has a variable as subscript, e.g.,  $\rmL_x$ for some $x$.

\smallskip

Thus, an $\rmL$-block in $\Tcar$ appears as $\rmL_x \, p \, x$ instead of $\rmL \, p \, n$. And so on. In particular, we do not repeat how binding works in $\Tcar$.

\smallskip

{\em Note that we assume that $\lambda$-trees in $\Tcar$ are closed and that in a $\lambda$-tree~$\bft$, the bound variables are different.}

\medskip

In the remainder of this Section we present a number of simple results about related facts, concerning $\Tcar$ and $\Tfre$.

\smallskip

Most importantly, there is a well-known {\em isomorphism\/} between $\beta$-reductions in $\Tcar$ and $\Tfre$.
See Lemma~\ref{LemComIso} below.
First, we define the mappings between $\Tcar$ and $\Tfre$ and vise versa.

\begin{Proc}\label{ProComFre}
Let $\bfs \in \Tcar$. Then we obtain $[\bfs] \in \Tfre$ by the following method.

$(i)$ Let $x$ be a variable in $\bfs$, bound in $\Tcar$ via the $\rmL$-block $\rmL_x \, p \, x$. Replace this $x$ by $\lVert p \lVert + 1$. Do this for all num-labels.

$(ii)$ Erase all subscripts, such as $x$, below labels $\rmL_x$ in $\bf s$.
\end{Proc}

\begin{Proc}\label{ProComCar}
Let $\bft \in \Tfre$. Then we obtain $\langle \bft \rangle \in \Tcar$ as follows.

For each label $\rmL \in \bft$, find all num-variables $n_1, \ldots, n_k$ bound by this $\rmL$. (If the mentioned $\rmL$ occurs as final label in the path $p \, \rmL \inw \bft$, then the bound num-variables occur in ${\it tree}(p \, \rmL)$.)

Now replace $\rmL$ by $\rmL_x$, using a {\em new\/} variable $x$ (i.e., a variable which has not yet been used in this procedure), and (if the number of bound $n_i$'s is not zero) replace each of these $n_i$ by $x$.
\end{Proc}

We extend $\alpha$-equivalence to $\Tfre$. We assume that the reader understands what `correspond' means in the following Definition and Lemmas.

\begin{Def}\label{DefComVar}
$(i)$ Let $\bft$ be a $\lambda$-tree in either $\Tcar$ or $\Tfre$. We obtain the {\it variable-free tree\/} of $\bft$ by stripping all variables and num-labels, including the subscripts of labels $\rmL_x \in \Tcar$. All other labels and all edges, including the edges that had a num-label as label, stay as they are.

$(ii)$ Let $\bft \in \Tcar$ and let $\bft'$ be a $\lambda$-tree in either $\Tcar$ or $\Tfre$. Then $\bft$ is {\em $\alpha$-equivalent\/} to $\bft'$ (in symbols: $\bft \equiv_\alpha \bft'$) if the variable-free trees of $\bft$ and $\bft'$ are identical and the bindings in $\bft$ correspond one to-one to the bindings in $\bft'$.
\end{Def}

\begin{Lem}\label{LemComInv}
 $(i)$ Mappings $[-]$ and $\langle - \rangle$ are each others inverses modulo alpha-conversion.

(ii) Let $\bfs \in \Tcar$ and $\bft \in \Tfre$. Then $\bfs \equiv_\alpha [\bfs]$ and $\bft \equiv_\alpha \langle \bft \rangle$.
\end{Lem}

Now we show that the mapping $[~\,]$ from $\Tcar$ to $\Tfre$ `preserves' the binding relation between variables and $\rmL$-labels.

\begin{Lem}\label{LemComIso}
Let $\bfs_1 \in \Tcar$ and $\bfs_1 \red \bfs_2$ by the $\beta$-reduction with pivot $L_x$. Then there is a corresponding $\beta$-reduction $[\bfs_1] \red [\bfs_2]$ in $\Tfre$, with corresponding pivot $\rmL$.

{\it  Proof}~~ In the reduction $\bfs_1 \red \bfs_2$ the bindings are preserved (common knowledge).
Moreover, we have $\bfs_1 \equiv_\alpha [\bfs_1]$ (Lemma~\ref{LemComInv}\,$(ii)$), and in the reduction $[\bfs_1] \red [\bfs_2]$ the bindings are preserved as well (Theorem~\ref{TheUsuUpd}). Note that the
variable-free trees of $\bfs_2$ and $[\bfs_2]$ are identical (follows from the construction procedures for $[\bfs_2]$). So $\bfs_2 \equiv_\alpha [\bfs_2]$.
\end{Lem}

Accordingly, there is a corresponding lemma for the inverted situation. We shall not go into it.

\section{Alternative beta-reductions}
\label{SecAlt}

\subsection{Balanced beta-reduction in namefree lambda-calculus}\label{SecBal}

There is a variant of $\beta$-reduction that is interesting if it is advantageous to keep all the information that is present in the original $\lambda$-calculus term. Then an argument ${\it tree}(p\, \rmS)$ should remain in the $\beta$-reduced term, just as the pivotal $\rmA$-$\rmL$-pair (see Definition~\ref{DefUsuGra} and \ref{DefUsuRed}). We  call this reduction relation {\em balanced $\beta$-reduction\/}. In the literature, it originally appeared under the name $\beta_1$ (\cite{Ned73}). For details, see the more recent literature on the Linear Substitution Calculus (cf.\ (\cite{AccKes}) and (\cite{AccKes12})), in which it is called {\em distant beta\/}, symbol  $\rightarrow_{dB}$. See also (\cite{BarBon}) and (\cite{KB05}).

\smallskip

We start with the definition of a {\em balanced path\/} in a $\lambda$-tree.

\begin{Def}\label{DefBalPat}
A path $p$ in a $\lambda$-tree ${\bf t}$ is called {\em balanced\/}, denoted ${\it bal}(p)$, if it is constructed by means of the following inductive rules:

$(i)$ ${\it bal}(\varepsilon)$, i.e., the empty string is balanced;

$(ii)$ if ${\it bal}(p)$, then ${\it bal}(\rmA \, p \, \rmL)$;

$(iii)$ if ${\it bal}(p)$ and ${\it bal}(q)$, then ${\it bal}(p \, q)$.

\noindent In case~$(ii)$, we say that the mentioned $\rmA$ {\it matches\/} the mentioned ${\rmL}$.
\end{Def}

Examples of balanced paths: $\varepsilon$, $\rmA \, \rmL$, $\rmA \, \rmA \, \rmL \, \rmL$, $\rmA \, \rmL \, \rmA \, \rmL$, $\rmA \, \rmA \, \rmL \, \rmA \, \rmA \, \rmL \, \rmL \, \rmL$.


Note the close correspondence between balanced paths and (consecutive) nested pairs of parentheses. Note that only $\rmA$- and $\rmL$-labels occur on balanced paths, so there is no other label involved, such as $\rmS$.

\smallskip

Now maintenance of the pivotal $\rmA$-$\rmL$-pair in $\lambda$-tree $\bft$, as mentioned above, has a serious consequence: it possibly prevents other instances of $\beta$-reduction, that arise in a `normal' $\beta$-reduction. If, for example, the underlined pair \underline{$\rmA \, \rmL$} in the path $p \, \rmA \, \underline{\rmA \, \rmL} \, \rmL \inw \bft$ is the pivotal pair, then the maintenance of this pair prevents the other $\rmA$ and the other $\rmL$ from appearing as a new pivotal pair after the one-step $\beta$-reduction induced by \underline{$\rmA \, \rmL$}. With `normal' $\beta$-reduction, this does not happen since \underline{$\rmA \, \rmL$} then disappears.

This situation can be avoided by using {\em balanced\/} $\beta$-reduction.

The following definition is an introduction to the notion `balanced reduction'.

\begin{Def}\label{DefActBal}
Let ${\bf t}$ be a $\lambda$-tree, let $b \in {\bf t}$ be a balanced path, and assume that $p \, \rmA \, b \, \rmL \in^\wedge {\bf t}$. This root path is called {\em active\/} if there is at least one path $p \, \rmA \, b \, \rmL \, q \, n \inwv {\bf t}$ such that $n$ is bound by $\rmL$ after $b$. If there is no such path, the root path is {\em inactive}.
\end{Def}

Now we give the `balanced' variant of $\beta$-reduction, with symbol $\rightarrow_b$. We recall that `$\rmL \, q \, n$ is an $\rmL$-block' is equivalent to `$n$ is bound by the initial $\rmL$\,'.

\begin{Def}\label{DefBalRed}
Let ${\bf t}$ be a $\lambda$-tree, let $b \in {\bf t}$ be a balanced path, and assume that $p \, \rmA \, b \, \rmL \in^\wedge {\bf t}$ is an active root path.
Let $\bft'$ be ${\bf t}$ in which all paths $p \, \rmA \, b \, \rmL \,
q \, n$ with $\rmL \, q \, n$ being an $\rmL$-block  have been
replaced by $p \, \rmA \, b \, \rmL \, q \, {\it tree}(p \, \rmS)$.
Then $\bft \rightarrow_b \bft'$.
\end{Def}

The condition that $p \, \rmA \, b \, \rmL$ is {\em active\/} in this definition avoids an infinite reduction path generated by the mentioned root path.

\begin{Def}\label{BalDefBlo}
$(i)$ The displayed $\rmL$ in Definition~\ref{DefBalRed} is called the {\em pivotal\/} $\rmL$.

$(ii)$ Let ${\bf t}$ be a $\lambda$-tree with $r = p \, \rmA \, b \, \rmL \, q \, n \in \bft$, where $b$ is a balanced path and such that the $\rmL \in r$ binds to $n$. We recall that $\rmL \, q \, n$ is called an $\rmL$-block (Definition~\ref{DefLamBou}). We call $\rmA \, b \, \rmL \, q \, n$ an $\rmA$-{\it block\/}.
\end{Def} 

Consider two $\lambda$-trees ${\bf t}$ and ${\bf t}'$ such that ${\bf t} \rightarrow_b {\bf t}'$ as described in Definition~\ref{DefBalRed}, so each $n$ bound by the pivotal $\rmL$ has been replaced by ${\it tree}(p \, S)$ in ${\it tree}(p \, \rmA \, b \, \rmL)$. Now we have that $\bft$ is (almost) a subtree of ${\bft}'$, provided that we omit all the num-variables $n$ in ${\bf t}$ bound by the pivotal $\rmL$ and omit the corresponding edges, as well. So, balanced $\beta$-reduction has the property that it {\it extends\/} the original underlying tree $\bft$, but for a number of num-variables that disappear.

\subsection{Focused beta-reduction in namefree lambda-calculus}\label{SecFoc}

Focused $\beta$-reduction is a special case of balanced $\beta$-reduction. This reduction concentrates on {\it precisely one\/} num-variable $n$ at a specific position in a certain $\lambda$-tree $\bft$, this $n$ being bound by a {\em pivotal\/} $\rmL$. Since $\rmL$ is pivotal, there must be an $\rmA$ `coupled' to $\rmL$, so $n$ is the final label of a particular path $p \, \rmA \, b \, \rmL \, q \, n \inwv \bft$, where $b$ is balanced. Focused $\beta$-reduction replaces this $n$ by the argument connected to the pivot. See the following definition, in which the symbol $\rightarrow_f$ is introduced for `focused' $\beta$-reduction.

\begin{Def}
Let $\bft$ be a $\lambda$-tree, let $b$ be a balanced path, and let $r = p \, \rmA \, b \, \rmL \, q \, n \inwv \bft$ be a fixed complete path in $\bft$. Let $\bft'$ be identical to $\bft$, except that $r$ has been replaced by $p \, \rmA \, b \, \rmL \, q \, {\it tree}(p \, \rmS)$.
Then $\bft \rightarrow_f \bft'$.
\end{Def}

It will be clear that we want a kind of $\beta$-reduction here that preserves the $\rmA$-$\rmL$-pair, because there may be other num-variables bound to this $\rmL$, and maybe one desires later to replace one or more of these by ${\it tree}(p \, \rmA)$, in subsequent $\rightarrow_f$-reductions.

\smallskip

The motivation for introducing focused $\beta$-reduction comes from the process known as {\em definition unfolding\/} in the {\it namecarrying\/} $\lambda$-calculus. Then a defined notion occurring in $M$, say, $x$, is replaced by the definiens, say, $P$. This action generally occurs for only one instance of the definiendum $x$. So instead of replacing {\em all\/} occurrences of $x$ in $M$, one aims at {\em precisely one\/} occurrence.

The `name' of the definiendum is important here, since it is hard to work with a `name-less' definiendum. Nevertheless, we address this variant of $\beta$-reduction here, since the name-less variant is interesting as such.

The possibility of having {\em balanced\/} $\beta$-reduction is necessary to be able to deal with other $\rmA$-$\rmL$-pairs, which otherwise would be inaccessible. See the following example.

\begin{Exa}\label{SubSecFocExa}
We have, in $\lambda$-calculus with normal untyped $\beta$-reduction:

$(\lambda x  \lamdot ((\lambda y  \lamdot M) Q)) P \red (\lambda x  \lamdot (M[y := Q])) P \red M[y := Q][x := P]$.

\noindent In {\em focused\/} $\beta$-reduction, this becomes:

$(\lambda x \lamdot ((\lambda y  \lamdot M) Q)) P \rightarrow_f (\lambda x  \lamdot (\lambda y  \lamdot M[y_0 := Q])Q) P \rightarrow_f$

\mbox{$(\lambda x  \lamdot ((\lambda y  \lamdot M[y_0 := Q])Q)[x_0 := P])P$}.

\noindent Here $y_0$ and $x_0$ are selected instances of the free $y$'s and $x$'s in $M$, respectively.

The second of the two one-step focused reductions would not be possible without the possibility of having a balanced $\lambda$-term $(\lambda y \ldots)Q$ between $\lambda x$ and $P$.

\end{Exa}

The following lemma is obvious.
\begin{Lem}
Let $\bft \in \Tfre$ and $\bft \rightarrow_b \bft'$. Then $\bft \,{\twoheadrightarrow_{f}}\, \bft'$.
\end{Lem}

\subsection{Erasing reduction}\label{SecEra}

After applying balanced or focused $\beta$-reduction, one also desires a reduction that removes the `remains', i.e., the $\rmA$ and the ${\rm L}$ in grafted trees $p \, \rmA \, b \, {\rm L} \, {\it tree}(p \, \rmA \, b \, \rmL)$ when no $n \in {\it tree}(p \, \rmA \, b \, \rmL)$ is bound to the displayed $\rmL$.
Moreover, the `argument' ${\it tree}(p \, \rmS)$ must be removed together with the mentioned $\rmS$. We call the corresponding reduction {\em erasing\/} reduction and use the symbol $\rightarrow_e$ for it. (This reduction is also referred to as `garbage collection' in the literature; see, e.g., (\cite{Ros93}).) 

The definition of $\rightarrow_e$ is not easy, because erasure applies to different parts of the original tree $\bft$.

\begin{Def}
Let ${\bf t}$ be a $\lambda$-tree and assume that a certain $p \, \rmA \, b \, {\rm L} \inw {\bf t}$, where $b$ is balanced. Moreover, assume that no num-variable in ${\it tree}(p \, \rmA \, b \, \rmL)$ is bound by the mentioned $\rmL$.

Then ${\bf t} ~\rightarrow_e~ {\bf t'}$, where $\bft'$ is $\bft$ in which $\rmS \, {\it tree}(p \, \rmS)$ has been removed and in which ${\it tree}(p)$ has been replaced by ${\it tree}(p \, \rmA)$ in which ${\it tree}(p \, \rmA \, b)$, in its turn, has been replaced by\\
$\{q \, n \inwv {\it tree}(p \, \rmA \, b \, \rmL) \, | \, n {\rm ~replaced~ by~} n-1 {\rm ~if~} n > \lVert q \lVert \}$.
\end{Def}

The necessity to replace $n$ by $n - 1$ is, of course, caused by the erasure of the mentioned $\rmL$.

Repeated application of $\rightarrow_e$ will result in a $\lambda$-tree without garbage.

\subsection{Theorems}

\begin{Lem} Reduction $~\rightarrow_e~$ is strongly normalizing, with a unique normal form.
\end{Lem}

We denote the reflexive, transitive closure of a reduction $\rightarrow_i$ by $\rrar_i$. An arbitrary sequence of reductions $\rightarrow_i$ and $\rightarrow_j$ is denoted $\rrar_{i,j}$.

\begin{The}\label{AltThe}
Let ${\bf t}$ and ${\bf t'}$ be $\lambda$-trees.

$(i)$ If ${\bf t} \rrar_\beta {\bf t}'$, then ${\bf t} \rrar_{b,e} {\bf t}'$.


$(ii)$ {\em (Postponement of $\rightarrow_e$ after $\rightarrow_b$)} ~ If ${\bf t} \rrar_{b,e} {\bf t'}$, then there is ${\bf t''}$ such that ${\bf t} \rrar_{b} {\bf t''} \rrar_e {\bf t'}$.

$(iii)$ {\em (Postponement of $\rightarrow_e$ after $\rightarrow_f$)} ~ If ${\bf t} \rrar_{f,e} {\bf t'}$, then there is ${\bf t''}$ such that ${\bf t} \rrar_{f} {\bf t''} \rrar_e {\bf t'}$.
\end{The}
{\it Proof\/} ~~ $(i)$ Easy.~
$(ii)$ (\cite{Ned73}, p.\ 48, Theorem 6.19).~
$(iii)$ Similarly. $\Box$

\begin{The}\label{TheConRed}
$\rightarrow_b$, $\rightarrow_f$ and $\rightarrow_e$ are confluent.
\end{The}
{\it Proof\/} For $\rightarrow_b$ and $\rightarrow_e$, see (\cite{Ned73}, Theorems 6.38 and 6.42). For $\rightarrow_f$, see (\cite{AccKes12}).






\section{A new, lossfree beta-reduction}\label{SecNewLos}
\label{SecLos}

\subsection{Expanding beta-reduction}\label{SecExp}

The system we are going to introduce takes a simpler approach than the one mentioned in Section~\ref{SecSho}.
Its syntax has a term node $m$, which we call {\em inner numeric label\/}
where $m \in \bbNp$.
An active $m$ holds the function $\tau_{0,m}$
, while a passive, i.e., \ present but ignored, $m$ holds the function $\tau_{0,0}$.
In some sense, we want to show that supporting the functions $\tau_{0,h}$
suffices to implement delayed updating in the basic namefree $\lambda$-calculus. 

\medskip

At the end of Section~\ref{SecBal} we mentioned that, when $\bft \rightarrow_b \bft'$, the $\lambda$-tree $\bft$ is almost a subtree of $\bft'$. The word `almost' concerns the fact that num-variables bound by the pivotal $\rmL \in \bft$ are removed in the balanced reduction, so they do not reappear in $\bft'$.

In the present section, we investigate what happens if we {\it leave the num-variables bound by $\rmL$ where they are}. In that case $\bft$ becomes a proper subtree of $\bft'$. We might say that the resulting reduction has the property that {\it no information from $\bft$ has been lost\/} in the reduction from $\bft$ to $\bft'$.

\smallskip

In order to make this work, we have to extend our notion of `path': now num-variables may appear everywhere {\it inside\/} a path, so not only at the end.

\begin{Def}
An {\em extended\/} path is a finite string of labels $\rmL$,  $\rmA$, $\rmS$ and arbitrary num-labels, ending in a num-label.
\end{Def}

Example: $\rmL \, \rmL \, 2 \, \rmA \, \rmA \, 1 \, \rmA \, \rmL \, \rmL \, 2$.

\smallskip

Now, num-labels come in two sorts: inside a path or at the end. We also obtain a new kind of $\lambda$-trees.

\begin{Def}\label{DefInnOut}
$(i)$ Num-variables not being end-labels, we call {\it inner num-labels\/}.
Num-variables that {\it are\/} end-labels (leaves), we refer to as {\it outer\/} num-labels.

$(ii)$ A $\lambda$-tree in which inner variables are allowed, we call an {\em extended\/} $\lambda$-tree.
\end{Def}

Consequently, the definition of a {\it balanced path\/} (Definition~\ref{DefBalPat}) must be adapted as well, such that it allows inner num-labels {\em inside\/} the string of ${\rm A}$'s and $\rmL$'s: from now on the notion `balanced path' will mean an `extended' one.

\smallskip

{\it In the remainder of this section, we assume that these new definitions of path, balanced path and $\lambda$-tree are valid. Moreover, we shall omit the word {\em extended} for the new paths and $\lambda$-trees.}

\smallskip

We recall from Section~\ref{SecCom} that the symbol ${\cal T}^{\it fre}$ concerns the set of namefree closed trees (without inner variables).
The set of namefree closed (extended) $\lambda$-trees where also inner variables are permitted is denoted by ${\cal T}^{\it exp}$.

Obviously, $\Tfre \subseteq \Texp$.

\smallskip

The $\beta$-like reduction being a consequence of this extension with inner variables, we call {\it expanding\/} $\beta$-reduction. Again, this reduction has two obvious flavors: {\it balanced\/} or {\it focused\/}. We concentrate from now on {\em focused\/} reduction, since this reduction can be used to simulate balanced reduction (cf Theorem~\ref{AltThe}\,$(ii)$).

We use the symbol `$\rightarrow_{\it ef}$' for expanding focused $\beta$-reduction. Its definition is as follows.

\begin{Def}\label{DefBetDel}
Let ${\bf t} \in {\cal T}^{\it exp}$, let $b \in {\bf t}$ be a balanced path (which now may contain inner variables), assume that $r = p \, \rmA \, b \, \rmL \, q \, n \inwv {\bf t}$ is a fixed, complete path in $\bft$, where $n$ is bound by the displayed~$\rmL$. Let $\bft'$ be identical to $\bft$, except that $r$ has been replaced by $p \, \rmA \, b \, \rmL \, q \, n \, {\it tree}(p \, \rmS)$.
Then $\bft \rightarrow_{\it ef} \bft'$.

\end{Def}


The effect of $\rightarrow_{\it ef}$-reduction is that the end-label $n$ {\it and} the edge labeled $n$ stay where they are, and ${\it tree}(p \, \rmS) $ is simply attached to this edge (recall our Note after Definition~\ref{DefLamTre}\,$(i)$). We shall see that the remaining presence of the label $n$ enables the update at a later stage.

For a pictorial representation, see Figure~\ref{PicDelRed}. Note: if ${\it tree}(p \, \rmS) $ consists only of a single edge, labeled with a num-variable, then this edge is just attached to the edge labeled $n$. 

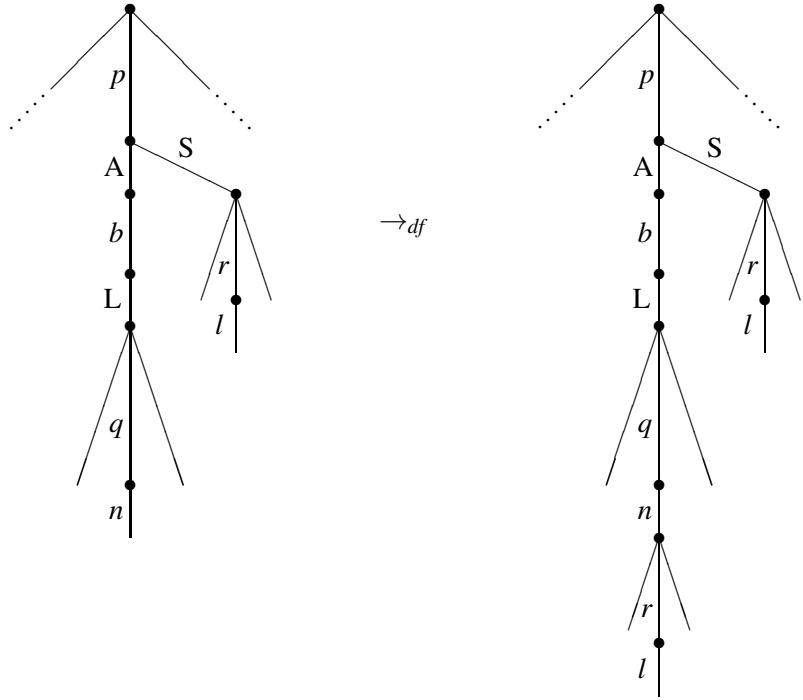
\begin{figure}[ht]

\begin{picture}(250,270)(0,10)

\put(50,70){\line(0,1){200}}
\put(50,270){\line(-1,-1){30}}
\put(50,270){\line(1,-1){30}}
\multiput(17,237.3)(-3,-3){5}{\circle*{1}}
\multiput(83,237.3)(3,-3){5}{\circle*{1}}

\put(50,220){\line(2,-1){40}}
\put(90,200){\line(0,-1){60}}

\put(90,200){\line(-1,-3){13.3}}
\put(90,200){\line(1,-3){13.3}}

\put(50,150){\line(-1,-3){20}}
\put(50,150){\line(1,-3){20}}

\put(50,90){\circle*{4}}
\put(50,150){\circle*{4}}
\put(50,170){\circle*{4}}
\put(50,200){\circle*{4}}
\put(50,220){\circle*{4}}
\put(50,270){\circle*{4}}
\put(90,200){\circle*{4}}
\put(90,160){\circle*{4}}

\put(42,77){$n$}
\put(42,110){$q$}
\put(40,157){$\rmL$}
\put(42,182){$b$}
\put(40,207){$\rmA$}
\put(42,242){$p$}
\put(68,214){$\rmS$}
\put(83,170){$r$}
\put(82,147){$l$}


\put(144,187){$\rightarrow_{\it df}$}

\put(250,270){\line(-1,-1){30}}
\put(250,270){\line(1,-1){30}}
\multiput(217,237.3)(-3,-3){5}{\circle*{1}}
\multiput(283,237.3)(3,-3){5}{\circle*{1}}

\put(250,220){\line(2,-1){40}}
\put(290,200){\line(0,-1){60}}

\put(290,200){\line(-1,-3){13.3}}
\put(290,200){\line(1,-3){13.3}}

\put(250,150){\line(-1,-3){20}}
\put(250,150){\line(1,-3){20}}

\put(250,90){\circle*{4}}
\put(250,150){\circle*{4}}
\put(250,170){\circle*{4}}
\put(250,200){\circle*{4}}
\put(250,220){\circle*{4}}
\put(250,270){\circle*{4}}
\put(290,200){\circle*{4}}
\put(290,160){\circle*{4}}

\put(242,77){$n$}
\put(242,110){$q$}
\put(240,157){$\rmL$}
\put(242,182){$b$}
\put(240,207){$\rmA$}
\put(242,242){$p$}
\put(268,214){$\rmS$}
\put(283,170){$r$}
\put(282,147){$l$}

\put(250,10){\line(0,1){260}}
\put(250,70){\circle*{4}}
\put(250,30){\circle*{4}}
\put(243,40){$r$}
\put(242,17){$l$}

\put(250,70){\line(-1,-3){11.7}}
\put(250,70){\line(1,-3){11.7}}

\end{picture}

\caption{A picture of namefree, expanding $\beta$-reduction}
\label{PicDelRed}
\end{figure}

\smallskip

We define what {\em inclusion\/} of (extended) $\lambda$-trees means.

\begin{Def}\label{DefIncTre}
Let ${\bf t}$ and ${\bf t}'$ be (extended) $\lambda$-trees. Then ${\bf t} \subseteq {\bf t'}$ if $p \inwv {\bf t}$ implies $p \inw {\bf t}'$. Moreover, ${\bf t} \subset {\bf t'}$ if ${\bf t} \subseteq {\bf t'}$ and $\bft \not\equiv \bft'$.
\end{Def}

\begin{The}\label{TheImpInc}
Let ${\bf t}, {\bf t'} \in \Texp$. Then ${\bf t} \rightarrow_{\it ef} {\bf t'}$ implies ${\bf t} \subset {\bf t'}$.
\end{The}

{\it Proof\/}~~Obvious.  ~~$\Box$

\begin{Lem}\label{LemBetExi}
Let ${\bf t} \in {\cal T}^{\it exp}$ and $p \, n \inw \bft$. Assume that $n$ is an {\em inner} num-variable. Then there is an $\rmL \in p$ that binds the $n$ and a matching $\rmA \in p$.
\end{Lem}

{\it Proof\/}~~An inner num-label $n$ can only appear in a $\rightarrow_{\it ef}$-reduction, when the $\rmL$ binding $n$ is pivotal in the reduction and this $\rmL$ is matched to an $\rmA$ (see definition~\ref{DefBetDel}). This situation is maintained by the expanding nature of $\rightarrow_{\it ef}$-reduction (Theorem~\ref{TheImpInc}). ~~$\Box$

\subsection{Tracing the binder in expanding beta-reduction}\label{SecTra}

Let ${\bf t}_0 \in {\cal T}^{\it exp}$ and ${\bf t}_0 \rrar_{\it ef} {\bf t}$, so ${\bf t} \in \Texp$ is the result of a series of expanding, focused reductions. These reductions may introduce inner variables, so it is not immediately clear what the binders are for (inner or outer) variables. In this section, we investigate how to determine the binder of a num-variable in ${\bf t}$.

\smallskip

Let $p \, n \inw {\bf t}$. Here, $n$ can be an inner or an outer num-label. We describe a pushdown automaton ${\cal P}$ that finds the {\em ${\rm L}$-binder} of $n$, i.e., the label ${\rm L} \in p$ that binds $n$ (this label always exists, since ${\cal T}^{\it exp}$ only contains closed terms).

\smallskip









We now present and explain the action of the pushdown automaton $\calP$. Let ${\bf t} \in \Texp$ and $p \, n \inw {\bf t}$. Assume that we desire to apply algorithm $\calP$ to find the ${\rm L}$-binder of $n$.

\begin{Rem} \label{RemTraPre} {\em Preliminary remarks.}

In algorithm $\calP$, we employ {\it states\/} that are pairs of natural numbers: $(k, l)$. We start with the insertion of a pair $(n,i)$  in the string $p \, n$, {\it between\/} $p$ and $n$. Here, $i$ originally is $0$ or $1$. The automaton moves the pair {\it to the left\/} through $p$, one step at a time, successively passing the labels in $p$ and meanwhile adapting the numbers in the pair.

The automaton has an outside {\it stack\/} that will contain certain states that are {\it pushed\/} at the top of the stack; a state on top of the stack can also be {\it popped back\/}, i.e., inserted into the path $p$, again.

The {\it transitions\/} are described in Definition~\ref{DefAlgBin}. A possible one-step transition is denoted by the symbol $\rightarrow$.

The procedure may be complicated by several recursive calls. In every recursive call, the algorithm starts with a `new' num-variable $j$ on the path $p$. Subsequently, it finds not only the $\rmL$ binding this $j$, but also {\em the $\rmA$ matching the~$j$}.
\end{Rem}


\smallskip

The formal description of $\calP$ is the following.

\begin{Proc}\label{DefAlgBin}

{\bf Preparation:} Transform $p \, n$ into $p \, (n,i) \, n$, where $i = 0$ if the goal is to find the ${\rm L}$-binder, and $i = 1$ in the recursion, when both the $\rmL$-binder of $n$ and the matching ${\rm A}$ are detected.

Now {\bf start} $\calP$ employing the following transition rules.

\smallskip

$(1)$ ~ {\it first step:} ${\it stack} = \emptyset$

$(2)$ ~ $p ~ {\rm L} ~ (m,k) ~ q ~ \rightarrow ~ p ~ (m \makebox{--} 1,k) ~ {\rm L} ~ q$, if $m > 0$

$(3)$ ~ $p ~ {\rm A} ~ (m,k) ~ q ~ \rightarrow ~ p ~ (m,k) ~ {\rm A} ~ q$, if $m > 0$

$(4)$ ~ $p ~ {\rm S} ~ (m,k) ~ q ~ \rightarrow ~ p ~ (m,k) ~ {\rm S} ~ q$, if $m > 0$

$(5)$ ~ $p ~ j ~ (m,k) ~ q ~ \rightarrow ~ p ~ (j,1) ~ j ~ q$, if $m > 0$; {\it push} $(m,k)$

$(6)$ ~ $p ~ {\rm L} ~ (0,l) ~ q ~ \rightarrow ~ p ~ (0,l \makebox{+} 1) ~ {\rm L} ~ q$, if $l > 0$

$(7)$ ~ $p ~ {\rm A} ~ (0,l) ~ q ~ \rightarrow ~ p ~ (0,l \makebox{--} 1) ~ {\rm A} ~ q$, if $l > 0$

$(8)$ ~ $p ~ j ~ (0,l) ~ q ~ \rightarrow ~ p ~ (0,l) ~ j ~ q$, if $l > 0$

$(9a)$ \, $p ~ (0,0)  ~ q ~ \rightarrow ~ p ~ {\it pop}  ~ q$, if  ${\it stack} \not = \emptyset$

$(9b)$ \, $p ~ (0,0) ~ q ~ \rightarrow {\it stop}$, if ${\it stack} = \emptyset$.

\end{Proc}


\begin{Lem}\label{LemALBlo}
Let ${\bf t} \in \Texp$ and $p \, n \inw {\bf t}$.  Apply $\calP$ to $p \, (n,i) \, n$, where $i = 0$ or $1$.

$(i)$ If $\calP$ stops in $p' \, (0,0) \, q' \, n$ with {\em empty\/} stack (see rule 9b), then each possible recursion has ended. Moreover, $q' \, n \equiv {\rm L} \, q'' \, n$, and the mentioned ${\rm L}$ binds the $n$.

$(ii)$ If $\calP$ stops in $p' \, (0,0) \, q' \, n$ with {\em non-empty\/} stack (see rule 9a), the $\rmA$ matching the binding $\rmL$ of $n$ has been found, the top-element of the stack is popped and $\calP$ continues where it had stopped before the recursion step.
\end{Lem}

{\em Proof\/} ~~ See Section~\ref{SecFur}.

\begin{Exa}\label{ExaProPat}
We look for the ${\rm L}$-binder of the final num-label, i.e., $3$, in the path ${\rm A \, L \, L \, L \, A \, L \, S \, 1 \, L \, 3} $. 

So we start $\calP$, inserting ${(3,0)}$ between the second to last label and the final label. Next, we obtain the following sequence of steps. (The superscripts to the arrows point at the number of the rule employed.)

\smallskip

${\rm A \, L \, L \, L \, A \, L \, S \, 1 \, L \, {\bf(3,0)} \, 3}$~
(${\it stack} = \emptyset$)~ $\rar^{(2)}$

${\rm A \, L \, L \, L \, A \, L \, S \, 1 \, {\bf(2,0)} \, L \, 3}$~
(${\it push~}  {\bf(2,0)}$)~ $\rar^{(5)}$

${\rm A \,L \, L \, L \, A \, L \, S \, {\bf(1,1)} \, 1 \, L \, 3}$~ $\rar^{(4)}$

${\rm A \, L \, L \, L \, A \, L \, {\bf(1,1)} \, S \, 1 \, L \, 3}$~ $\rar^{(2)}$

${\rm A \, L \, L \, L \, A \, {\bf(0,1)} \, L \, S \, 1 \,  L \, 3}$~ $\rar^{(7)}$

${\rm A \, L \, L \, L \, {\bf(0,0)} \, A \, L \, S \, 1 \,  L \, 3}$~  (${\it pop~} {\bf(2,0)}$)~ $\rar^{(9a)}$

${\rm A \, L \, L \, L \, {\bf(2,0)} \, A \, L \, S \, 1 \,  L \, 3}$~   $\rar^{(2)}$

${\rm A \, L \, L \, {\bf(1,0)} \, L \, A \, L \, S \, 1 \,  L \, 3}$~  $\rar^{(2)}$

${\rm A \, L \, {\bf(0,0)} \, \underline{\rule[-0.1em]{0em}{1em} L \,  L \, A \, L \, S \, 1 \, L \, 3}}$~\,
$\rar^{(9b)}$ (${\it stack} = \emptyset$: {\it stop})

 \hspace{1.5cm} {\rm L}-block of $3$
\end{Exa}

It now follows that the final label 3 on the path ${\rm A \, L \, L \, L \, A \, L \, S \, 1 \, L \, 3} $ is bound to the second label $\rmL$ on the left.

\section{Further Work and Acknowledgements}\label{SecFur}

There is much more to be said about the notion of expanding $\beta$-reduction as presented in Section~\ref{SecExp}. We intend to do that in a forthcoming paper, including conventional proofs, theorems on expanding $\beta$-reduction, and a connection between weak and strong normalization for this reduction.

Since 2021, one of the authors (Ferruccio Guidi) is formalizing and checking all proofs by means of the theorem prover Matita (\cite{Asp}).

----------------------------------------------------------------------------------------------------------------

I, Rob Nederpelt, express special thanks to Vincent van Oostrom for his interest in a pre-version of this paper. I also thank Herman Geuvers for encouraging remarks and a thorough review of an earlier version of this paper.

I, Ferruccio Guidi, would like to dedicate the results presented in these pages and those to come to Anyelis Marielbys Parra Pire, a special friend whose constant closeness accompanied me in the development of this work.

Both authors thank anonymous referees for their careful reading, leading to corrections and improvements.

\bibliographystyle{eptcs}
\bibliography{festschrift}

\end{document}